\documentclass{article}

\usepackage{arxiv}
\usepackage[utf8]{inputenc}
\usepackage[T1]{fontenc}

\usepackage{xcolor}
\usepackage{amsfonts,amsthm}
\usepackage{amssymb}
\usepackage{amsmath}
\usepackage[round]{natbib}
\usepackage{dutchcal}
\usepackage{bm}
\usepackage{graphicx}
\usepackage{booktabs,array}
\usepackage{caption}
\usepackage{subcaption}
\usepackage{enumerate}
\usepackage{tikz}
\usepackage[colorlinks=true,linkcolor=blue,citecolor=blue,urlcolor=blue]{hyperref}
\IfFileExists{xurl.sty}{\usepackage{xurl}}{}

\newcommand{\CRMTitle}{The Geometry of Existence and Uniqueness of Maximum Likelihood Estimation in Categorical Response Models}
\newcommand{\CRMArxivKeywords}{Separation \and latent polyhedra \and structure vectors \and log-concavity \and linear programming}
\newcommand{\CRMPDFKeywords}{Separation, Latent Polyhedra, Structure Vectors, Log-Concavity, Linear Programming}
\newcommand{\CRMPDFAuthor}{Lukas Sablica; Kurt Hornik; Thomas Rusch}

\newcommand{\pkg}[1]{{\normalfont\fontseries{b}\selectfont #1}}
\newcommand{\R}{\mathbb{R}}
\renewcommand{\P}{\mathbb{P}}
\newcommand{\nS}{q}

\theoremstyle{plain}
\newtheorem{theorem}{Theorem}
\newtheorem{lemma}{Lemma}
\newtheorem{corollary}{Corollary}
\newtheorem{prop}{Proposition}
\theoremstyle{definition}
\newtheorem{definition}{Definition}

\theoremstyle{remark}

\hypersetup{
  pdftitle={\CRMTitle},
  pdfauthor={\CRMPDFAuthor},
  pdfkeywords={\CRMPDFKeywords}}

\title{\CRMTitle}

\author{
  Lukas Sablica \\
  Institute for Statistics and Mathematics\\
  WU Vienna University of Economics and Business\\
  Austria\\
  \url{https://www.wu.ac.at/en/statmath}\\
  ORCiD: 0000-0001-9166-4563\\
  \href{mailto:Lukas.Sablica@wu.ac.at}{\nolinkurl{Lukas.Sablica@wu.ac.at}}
   \And
 Kurt Hornik\\
Institute for Statistics and Mathematics\\
WU Vienna University of Economics and Business\\
Austria\\
\url{https://www.wu.ac.at/en/statmath}\\
ORCiD: 0000-0003-4198-9911\\
\href{mailto:Kurt.Hornik@wu.ac.at}{\nolinkurl{Kurt.Hornik@wu.ac.at}}
\And
 Thomas Rusch\\
Competence Center for Empirical Research Methods\\
WU Vienna University of Economics and Business\\
Austria\\
\url{https://www.wu.ac.at/en/competence-center-for-empirical-research-methods}\\
ORCiD: 0000-0002-7773-2096\\
\href{mailto:thomas.rusch@wu.ac.at}{\nolinkurl{Thomas.Rusch@wu.ac.at}}
}

\begin{document}
\maketitle

\begin{abstract}
Nonexistence of the maximum likelihood estimate (MLE) under separation is treated as a solved problem for binary logistic regression and as a scattered collection of model-specific results everywhere else. We show that it is one phenomenon with one criterion. In a latent polyhedral categorical response model, every observed outcome corresponds to a polyhedral event in latent variables whose faces shift linearly with the parameter. Random-utility choice, cumulative-link, ranking, multivariate binary and ordinal, sequential, adjacent-category logit, and fixed-score stereotype models belong to this family. The likelihood sees each observed factor only through the columns of its threshold map, the structure vectors. A finite MLE exists if and only if the pooled structure vector set has overlap. Sufficiency requires only continuity, necessity requires only strictly threshold-increasing probabilities, and neither likelihood concavity, exchangeability, nor full design rank is needed. Positive strictly log-concave latent densities and threshold-identifiable polyhedra then yield uniqueness on the estimable span. A single linear program determines whether overlap holds, and convex cone geometry measures the dimension of separation. An application using cumulative-link models for willingness to share health data identifies observations and model terms associated with nonexistence and illustrates how the diagnostics inform model revision and sensitivity analysis.
\end{abstract}


\section[Introduction]{Introduction}

When a categorical regression returns enormous coefficients, a singular Hessian, or enormous or missing standard errors, the usual suspects are the optimizer or the starting values. In binary logistic regression, separation is a familiar alternative explanation, for which there is a complete theory and growing software support \citep{detectseparation}. Separation of the observed covariate configurations makes a finite maximum likelihood estimate (MLE) impossible, and overlap makes it exist \citep{polytomous:Silvapulle:1981,polytomous:Albert+Anderson:1984}. The same symptoms can also arise in cumulative-link, multinomial, ranking, and multivariate latent-response models, which together carry much of applied categorical data analysis \citep{polytomous:Fahrmeir+Tutz:1994}. There the phenomenon is much less understood and symptoms can again be mistaken for numerical difficulties. This paper shows that this interpretation misses the structure of the problem. Across all of these models, nonexistence of the MLE is a single geometric property of the data and the model. One criterion characterizes existence, and one linear program establishes it.

The classical theory connects nonexistence of the MLE to separation of the observed predictor configurations \citep{polytomous:Silvapulle:1981,polytomous:Albert+Anderson:1984}. It was then extended model by model, to quantal and ordinal regression \citep{polytomous:Tse:1986,polytomous:Kaufmann:1988}, to multinomial and multivariate probit \citep{probit:Lesaffre-Kaufmann:1992,polytomous:Speckman+Lee+Sun:2009}, and to boundary problems in log-linear and exponential-family models \citep{Fienberg+Rinaldo:2012}. Each extension established existence for one class, and each leaned on properties that class happened to possess, such as concavity of the likelihood, full design rank, exchangeability, or one particular latent-error construction. Our existence criterion  abstracts from the specific narrow model class and requires neither likelihood concavity, exchangeability, nor full design rank.

That statement follows from an observation that is elementary once made. Write the event that observation $i$ falls in category $j$ as a polyhedron in latent variables,
\[
\left\{\bm{u}:H_{ij}\bm{u}\leq B_{ij}'\bm{\beta}\right\},
\]
with a fixed matrix $H_{ij}$ and bounds $B_{ij}'\bm{\beta}$ linear in the parameter $\bm{\beta}$. The events in random-utility choices are systems of utility comparisons, while in cumulative-link outcomes they are intervals.
Rankings are ordered comparison events, multivariate binary patterns are orthants, and multivariate ordinal responses are rectangles.
Sequential outcomes are signed systems of stagewise threshold inequalities. Adjacent-category logit models and fixed-score stereotype models enter through exact choice representations.
We call this family \emph{latent polyhedral categorical response models}. The likelihood sees each contribution from an observed category only through the columns of $B_{ij}$, which we call \emph{structure vectors}. Every result in this paper is a statement about their pooled set. The family is delimited by the linearity of the thresholds. As an example, reference-ratio models \citep{peyhardi+trottier+guedon:2015,glmcat}, whose comparison thresholds are nonlinear in the parameter, fall outside it.

Our first result is the existence criterion. The pooled structure vectors have overlap if no nonzero direction in their span weakly enlarges every latent event corresponding to an observed category. Under overlap, continuity of the latent threshold probabilities alone forces the likelihood to attain its maximum at a finite point. If positive-probability events also respond strictly to every effective threshold increase, overlap is necessary. For standard choice-type and cumulative-link models with positive error densities this gives an if and only if characterization. Rank-deficient designs require no reparameterization. Existence is settled on the span of the structure vectors, the estimable span, and the likelihood is constant orthogonally to it.

Our second result concerns uniqueness. A positive and strictly log-concave latent density produces strictly log-concave polyhedral probabilities whenever distinct admissible thresholds define genuinely different latent events, a property we call threshold identifiability. Together with the spanning behavior of the structure vectors, this makes the log-likelihood strictly concave on the estimable span. Examples include multinomial logit and probit, cumulative logit and probit, random-utility rankings, and multivariate binary and ordinal probit, with dependent and nonexchangeable latent errors allowed throughout.

The framework also gives a reusable diagnostic workflow for routinely checking separation. An analyst can write the observed categories as latent inequalities, pool the threshold columns, and check for overlap by linear programming. Using the same geometry, they can determine the separation type, describe the recession cone and identify implicated observations or model terms. This separates structural failure of the likelihood from numerical failure of an algorithm and provides a principled basis for model revision or penalized estimation. Penalized and bias-reduced procedures provide finite alternatives \citep{firth,kosmidis-lunardon-2024}, but do not determine whether the ordinary MLE exists or where a failure originates.

Section~\ref{sec:framework} introduces the family and the separation geometry, with its principal members collected in Table~\ref{tab:family}. Section~\ref{sec:existence} develops the overlap criterion and its model consequences. Section~\ref{sec:diagnostics} presents the linear-programming and cone-geometric diagnostics. Section~\ref{sec:uniqueness} establishes uniqueness, and Section~\ref{sec:application} illustrates the diagnostic workflow with a cumulative-link model for willingness to share health data. Proofs are given immediately after the corresponding results. Additional details for the empirical illustration are collected in the Appendix.

\section[Latent Polyhedral Models and Separation]{Latent Polyhedral Models and Separation}
\label{sec:framework}

\subsection[Likelihood and Latent Polyhedral Representation]{Likelihood and Latent Polyhedral Representation}

Write $[m]=\{1,\ldots,m\}$ and $[m]_{-j}=[m]\setminus\{j\}$. Consider $n$ independent multinomial count vectors $\bm{y}_i=(y_{i1},\ldots,y_{im})'$ with probability vectors $\bm{\pi}_i=(\pi_{i1},\ldots,\pi_{im})'$ and totals $n_i=\sum_{j=1}^m y_{ij}\geq1$. Let $\bm{\beta}\in\R^p$ be the unknown parameter. We write
\[
\bm{\pi}_i=h(Z_i'\bm{\beta})
=\bigl(h_1(Z_i'\bm{\beta}),\ldots,h_m(Z_i'\bm{\beta})\bigr)',
\]
where $Z_i=[\bm{z}_{i1},\ldots,\bm{z}_{im}]\in\R^{p\times m}$ is known and $h$ maps the linear predictors to the probability simplex. When a baseline category is used, we impose $\bm{z}_{im}=\bm{0}$ and write $Z=[Z_1,\ldots,Z_n]$. With $I_j=\{i:y_{ij}>0\}$, the likelihood and log-likelihood, up to terms independent of $\bm{\beta}$, are
\begin{equation}
\label{eq:likelihood}
L(\bm{\beta})
=\prod_{i=1}^n\prod_{j=1}^m h_j(Z_i'\bm{\beta})^{y_{ij}}
=\prod_{j=1}^m\prod_{i\in I_j}h_j(Z_i'\bm{\beta})^{y_{ij}},
\qquad
\ell(\bm{\beta})=\sum_{j=1}^m\sum_{i\in I_j}y_{ij}\log h_j(Z_i'\bm{\beta}).
\end{equation}

\begin{definition}[Latent polyhedral categorical response model]
For observation $i$, let $\bm{U}_i\in\R^{r_i}$ have a distribution that does not depend on $\bm{\beta}$. For each category $j\in[m]$, let $q_{ij}\geq1$, let $H_{ij}\in\R^{q_{ij}\times r_i}$ and let
$B_{ij}=[\bm{b}_{ij1},\ldots,\bm{b}_{ijq_{ij}}]\in\R^{p\times q_{ij}}$ be known. The model is a \emph{latent polyhedral categorical response model} if, for every admissible $\bm{\beta}$ and up to latent events of probability zero, the sets
\begin{equation}
\label{eq:polyhedral-event}
\mathcal{C}_{ij}(\bm{\beta})
=\left\{\bm{u}\in\R^{r_i}:H_{ij}\bm{u}\leq B_{ij}'\bm{\beta}\right\},
\qquad j\in[m],
\end{equation}
form a partition of the latent space and
\begin{equation}
\label{eq:linear-threshold}
h_j(Z_i'\bm{\beta})
=\P\bigl(\bm{U}_i\in\mathcal{C}_{ij}(\bm{\beta})\bigr).
\end{equation}
The inequalities are coordinatewise. Since the thresholds $B_{ij}'\bm{\beta}$ are linear in $\bm{\beta}$, we also call these linear-threshold categorical response models.
\end{definition}

Set $\bm{W}_{ij}=H_{ij}\bm{U}_i$ and let $G_{ij}$ be its CDF. The response probability can equivalently be written as
\begin{equation}
\label{eq:linear-threshold-cdf}
h_j(Z_i'\bm{\beta})
=\P\bigl(\bm{W}_{ij}\leq B_{ij}'\bm{\beta}\bigr)
=G_{ij}(B_{ij}'\bm{\beta}).
\end{equation}
The columns $\bm{b}_{ijr}$ of $B_{ij}$ are the \emph{structure vectors}. Only factors with positive counts affect the likelihood, so their pooled structure set is
\begin{equation}
\label{eq:S}
\mathcal{S}
=\left\{\bm{b}_{ijr}:j\in[m],\ i\in I_j,\ r\in[q_{ij}]\right\}.
\end{equation}
The matrices $H_{ij}$ determine which latent comparisons or bounds define an outcome. The vectors in $\mathcal{S}$ determine how the corresponding faces move with $\bm{\beta}$. This separation of roles lets one common geometry cover models whose conventional probability formulas look quite different. Table~\ref{tab:family} collects the members developed in this section, the latent events that represent their outcomes, the structure vectors that those events contribute, and the results that settle existence and uniqueness for them.

A latent event may have several equivalent inequality descriptions. The results depend on the pooled structure vectors only through the cone \(K=\operatorname{cone}(\mathcal S)\). Permuting, duplicating, or positively rescaling inequalities does not affect overlap, the estimable span, the lineality space, or the cone of separating directions. The same applies when adding or removing conically redundant structure vectors, as \(K\) remains unchanged. We therefore use effective representations and omit probability-irrelevant or conically redundant inequalities. For example, adding comparisons that follow transitively from an irredundant ranking representation does not change the structure cone.

\subsection{Choice-Type Models}
\label{sec:choice-models}

Choice-type models arise from latent utilities $\bm{z}_{ij}'\bm{\beta}+\xi_{ij}$ \citep{polytomous:Fahrmeir+Tutz:1994}. Category $j$ is generated when its utility exceeds that of every competitor. Rearranging the comparisons gives
\begin{equation}
\label{eq:G_j}
h_j(Z_i'\bm{\beta})
=\P\left(\xi_{ik}-\xi_{ij}\leq(\bm{z}_{ij}-\bm{z}_{ik})'\bm{\beta},\ \forall k\in[m]_{-j}\right)
=G_j\left(\bigl((\bm{z}_{ij}-\bm{z}_{ik})'\bm{\beta},\ k\in[m]_{-j}\bigr)\right).
\end{equation}
Here $\bm{U}_i=\bm{\xi}_i$ and $H_{ij}=H_j$, where the row of $H_j$ indexed by $k\in[m]_{-j}$ is $\bm{e}_k'-\bm{e}_j'$. The matching column of $B_{ij}$ is $\bm{z}_{ij}-\bm{z}_{ik}$. Hence the pooled structure set is
\begin{equation}
\label{eq:S-CT}
\mathcal{S}_{\mathrm{CT}}
=\left\{\bm{z}_{ij}-\bm{z}_{ik}:j\in[m],\ i\in I_j,\ k\in[m]_{-j}\right\}.
\end{equation}
Independent Gumbel or Gaussian utility errors yield multinomial logit or multinomial probit. More generally, neither independence nor exchangeability is required for the existence argument below.

The following proposition connects the geometry in Equation~\ref{eq:S-CT} to the usual design rank.

\begin{prop}
\label{lem:rank-equivalence}
Let $Z_i=[\bm{z}_{i1},\dots,\bm{z}_{im}]\in\mathbb{R}^{p\times m}$ and suppose that $\bm{z}_{im}=\bm{0}$ for every $i$. Then
\[
\mathrm{span}(\mathcal{S}_{\mathrm{CT}})=\mathrm{span}(Z),
\]
where $Z=[Z_1,\dots,Z_n]$.
\end{prop}

\begin{proof}
First, $\mathrm{span}(\mathcal{S}_{\mathrm{CT}})\subseteq\mathrm{span}(Z)$ holds because each $\bm{s}\in\mathcal{S}_{\mathrm{CT}}$ is a difference of two columns of $Z$, hence a linear combination of columns of $Z$.

For the reverse inclusion, fix any $i$. Since
$n_i=\sum_{j=1}^m y_{ij}\geq 1$, there exists some $j$ with
$y_{ij}>0$, i.e., $i\in I_j$.

If $j\neq m$, then under the baseline-$m$ constraint
$\bm{z}_{im}=\bm{0}$,
\[
\bm{z}_{ij}=\bm{z}_{ij}-\bm{z}_{im}\in\mathcal{S}_{\mathrm{CT}},
\]
so $\bm{z}_{ij}\in\mathrm{span}(\mathcal{S}_{\mathrm{CT}})$. For any
$k\in[m]_{-j}$,
\[
\bm{z}_{ik}=\bm{z}_{ij}-(\bm{z}_{ij}-\bm{z}_{ik}),
\]
and $\bm{z}_{ij}-\bm{z}_{ik}\in\mathcal{S}_{\mathrm{CT}}$ by construction.
Hence $\bm{z}_{ik}\in\mathrm{span}(\mathcal{S}_{\mathrm{CT}})$ for every $k$.

If $j=m$, then for every $k\in[m]_{-m}$,
\[
\bm{z}_{im}-\bm{z}_{ik}=-\bm{z}_{ik}\in\mathcal{S}_{\mathrm{CT}}.
\]
Therefore $\bm{z}_{ik}\in\mathrm{span}(\mathcal{S}_{\mathrm{CT}})$ for every
$k\neq m$, while $\bm{z}_{im}=\bm{0}$ belongs to this span as well.

Thus, in either case, all $\{\bm{z}_{ik}:k\in[m]\}$ lie in
$\mathrm{span}(\mathcal{S}_{\mathrm{CT}})$, and therefore
$\mathrm{span}(Z)\subseteq\mathrm{span}(\mathcal{S}_{\mathrm{CT}})$.

Combining both inclusions yields $\mathrm{span}(\mathcal{S}_{\mathrm{CT}})=\mathrm{span}(Z)$.
\end{proof}

When $m=2$, write $\bm{x}_i=\bm{z}_{i1}-\bm{z}_{i2}$ and let $F$ be the CDF of $\xi_{i2}-\xi_{i1}$. The response probabilities are $F(\bm{x}_i'\bm{\beta})$ and $1-F(\bm{x}_i'\bm{\beta})$. The signed structure vectors are $\bm{x}_i$ for observations in category 1 and $-\bm{x}_i$ for observations in category 2. Standard choices of $F$ give logit, probit, cloglog, and cauchit regression.

\subsection{Cumulative-Link Models}
\label{sec:cumulative-models}

For a cumulative-link ordinal model \citep{McCullagh,polytomous:Agresti:2010}, let $Y_i^*=\bm{x}_i'\bm{\theta}+\epsilon_i$ and let
$-\infty=\alpha_0<\alpha_1<\cdots<\alpha_{m-1}<\alpha_m=\infty$. The event $Y_i=j$ is
\begin{equation}
\label{eq:cumulative-interval}
\alpha_{j-1}-\bm{x}_i'\bm{\theta}<\epsilon_i
\leq\alpha_j-\bm{x}_i'\bm{\theta}.
\end{equation}
Consequently,
\[
\P(Y_i=j)=F(\alpha_j-\bm{x}_i'\bm{\theta})
-F(\alpha_{j-1}-\bm{x}_i'\bm{\theta}).
\]
Write $\bm{\beta}=(\alpha_1,\ldots,\alpha_{m-1},\bm{\theta}')'$ and let $\bm{e}_j$ be the $j$th coordinate vector in $\R^{m-1}$. The latent polyhedral matrices are
\begin{align*}
H_{i1}&=[1],
& B_{i1}&=[(\bm{e}_1',-\bm{x}_i')'],\\
H_{ij}&=(1,-1)',
& B_{ij}&=[(\bm{e}_j',-\bm{x}_i')',(-\bm{e}_{j-1}',\bm{x}_i')'],\quad 1<j<m,\\
H_{im}&=[-1],
& B_{im}&=[(-\bm{e}_{m-1}',\bm{x}_i')'].
\end{align*}
For an interior category, $H_{ij}\epsilon_i=(\epsilon_i,-\epsilon_i)'$ and the two thresholds impose the upper and lower endpoints of the interval. The transformed vector is singular in $\R^2$, which causes no problem because the primitive event remains a one-dimensional interval. The pooled structure set is
\begin{align}
\label{eq:S-CL}
\mathcal{S}_{\mathrm{CL}}&=\mathcal{S}_{\mathrm{CL},1}\cup\mathcal{S}_{\mathrm{CL},2},\\
\mathcal{S}_{\mathrm{CL},1}&=\left\{(\bm{e}_j',-\bm{x}_i')':j=1,\ldots,m-1,\ i\in I_j\right\},\\
\mathcal{S}_{\mathrm{CL},2}&=\left\{(-\bm{e}_{j-1}',\bm{x}_i')':j=2,\ldots,m,\ i\in I_j\right\}.
\end{align}
The vectors in $\mathcal{S}_{\mathrm{CL},1}$ move upper endpoints upward and those in $\mathcal{S}_{\mathrm{CL},2}$ move lower endpoints downward. When $m=2$, the cumulative and choice representations coincide. Logistic, Gaussian, and minimum extreme-value errors give cumulative logit, probit, and cloglog models.

\subsection{Further Direct Members of the Family}
\label{sec:direct-members}

The same construction covers outcomes that impose several latent comparisons or coordinate bounds at once.

\noindent\textbf{Rankings and partially observed comparisons.}
Ranked outcomes can be represented as systems of pairwise latent-utility
comparisons. Under independent extreme-value errors, this construction yields
the rank-ordered logit model of \citet{hausman_specifying_1987}, while the
formulation below allows general continuous joint error distributions and
partially observed comparison sets. Let $\mathcal{E}_i$ contain the ordered pairs $(a,b)$ for which observation $i$ reports alternative $a$ above alternative $b$. The observed event is
\begin{equation}
\label{eq:ranking-event-general}
\xi_{ib}-\xi_{ia}
\leq
(\bm{z}_{ia}-\bm{z}_{ib})'\bm{\beta},
\qquad (a,b)\in\mathcal{E}_i,
\end{equation}
with pooled structure set
\begin{equation}
\label{eq:ranking-structure-set}
\mathcal{S}_{\mathrm{rank}}
=\{\bm{z}_{ia}-\bm{z}_{ib}:(a,b)\in\mathcal{E}_i,\ i=1,\ldots,n\}.
\end{equation}
For a complete ranking $r_{i1}\succ\cdots\succ r_{iJ}$, the $J-1$ adjacent comparisons suffice. Partial rankings retain only reported comparisons. Top-$k$ and best-worst observations are represented by the comparisons implied by the reported sets or alternatives.

\noindent\textbf{Multivariate binary latent responses.}
For a multivariate binary latent responses \citep{probit:Lesaffre-Kaufmann:1992}, let $\bm{Y}_i=(Y_{i1},\ldots,Y_{id})'$ satisfy
$Y_{ir}=1\{\bm{x}_{ir}'\bm{\theta}+\epsilon_{ir}>0\}$ and let $X_i$ have rows $\bm{x}_{ir}'$. With
$D_i=\mathrm{diag}(2Y_{i1}-1,\ldots,2Y_{id}-1)$, the observed response pattern is the orthant event
\begin{equation}
\label{eq:multibinary-event-consequence}
-D_i\bm{\epsilon}_i\leq D_iX_i\bm{\theta}.
\end{equation}
The structure vector associated with component $r$ is $(2Y_{ir}-1)\bm{x}_{ir}$. This construction includes correlated multivariate binary-response models, and with a fixed positive-definite Gaussian covariance matrix, the bivariate or multivariate probit model of \citet{probit:Lesaffre-Kaufmann:1992}.

\noindent\textbf{Multivariate ordinal latent responses.}
For multivariate ordinal latent response \citep{li_likelihood_2008}, suppose component $r$ is observed in category $c$ when
$\alpha_{r,c-1}<\bm{x}_{ir}'\bm{\theta}+\epsilon_{ir}\leq\alpha_{r,c}$. Collect the lower and upper cutpoints selected by the observed pattern in $\bm{l}_i$ and $\bm{u}_i$. The joint response is the rectangle event
\begin{equation}
\label{eq:multiordinal-event-consequence}
\begin{pmatrix}I_d\\-I_d\end{pmatrix}\bm{\epsilon}_i
\leq
\begin{pmatrix}
\bm{u}_i-X_i\bm{\theta}\\
-\bm{l}_i+X_i\bm{\theta}
\end{pmatrix}.
\end{equation}
If cutpoints are estimated, their coordinate vectors enter the threshold maps exactly as in the univariate cumulative model. Boundary categories simply omit one side of the corresponding coordinate bound.

\noindent\textbf{Sequential responses.}
A sequential or continuation-ratio model \citep{polytomous:Tutz:1991} specifies
\[
\P(Y_i=k\mid Y_i\geq k)=G_k(\gamma_{ik}),
\qquad k=1,\ldots,m-1,
\]
where each predictor $\gamma_{ik}=\bm{z}_{ik}'\bm{\beta}$ is linear in the common parameter vector after any stage-specific blocks have been embedded in $\bm{\beta}$. Introduce independent stage errors $\epsilon_{ik}$ with CDF $G_k$ and write $\bm{U}_i=(\epsilon_{i1},\ldots,\epsilon_{i,m-1})'$. Observing category $j<m$ means continuing through stages $1,\ldots,j-1$ and stopping at stage $j$. Up to probability-zero boundaries, this event is
\begin{equation}
\label{eq:sequential-polyhedral}
H_j^{\mathrm{seq}}\bm{U}_i\leq (B_{ij}^{\mathrm{seq}})'\bm{\beta},
\qquad
H_j^{\mathrm{seq}}=
\begin{pmatrix}
-\bm{e}_1'\\
\vdots\\
-\bm{e}_{j-1}'\\
\bm{e}_j'
\end{pmatrix},
\qquad
(B_{ij}^{\mathrm{seq}})'\bm{\beta}=
\begin{pmatrix}
-\gamma_{i1}\\
\vdots\\
-\gamma_{i,j-1}\\
\gamma_{ij}
\end{pmatrix},
\end{equation}
with the continuation block absent when $j=1$. For the final category,
$H_m^{\mathrm{seq}}=-I_{m-1}$ and $(B_{im}^{\mathrm{seq}})'\bm{\beta}=(-\gamma_{i1},\ldots,-\gamma_{i,m-1})'$. Thus the structure vectors are $-\bm{z}_{ik}$ for continued stages and $\bm{z}_{ij}$ for the stopping stage. Independence gives the familiar staged binary-choice probabilities
\begin{equation}
\label{eq:sequential-probabilities}
\P(Y_i=j)=G_j(\gamma_{ij})\prod_{k<j}\{1-G_k(\gamma_{ik})\},
\quad j<m,
\qquad
\P(Y_i=m)=\prod_{k=1}^{m-1}\{1-G_k(\gamma_{ik})\}.
\end{equation}
Equation~\ref{eq:sequential-polyhedral} shows that the model is a direct latent polyhedral member, while Equation~\ref{eq:sequential-probabilities} gives its usual interpretation as a sequence of binary choices.

\subsection{Ordinal Models with Choice Representations}
\label{sec:choice-ordinal}

Two familiar ordinal logit constructions become latent polyhedral models through exact choice representations. The corresponding effective design vectors are explicit.

\noindent\textbf{Adjacent-category logit models.}
An adjacent-category logit model \citep{polytomous:Agresti:2010} specifies
\begin{equation}
\label{eq:adj_cat_def_simple}
\log\frac{\P(Y_i=j)}{\P(Y_i=j+1)}=\gamma_{ij},
\qquad j=1,\ldots,m-1.
\end{equation}
From Equation~\ref{eq:adj_cat_def_simple},
$\P(Y_i=j)=\P(Y_i=j+1)\exp(\gamma_{ij})$. Recursing to the final
category gives $\P(Y_i=j)=\P(Y_i=m)\exp(\Gamma_{ij})$, where
$\Gamma_{ij}=\sum_{k=j}^{m-1}\gamma_{ik}$ for $j<m$ and
$\Gamma_{im}=0$. Since the probabilities sum to one,
\begin{equation}
\label{eq:adjacent-softmax-consequence}
\P(Y_i=j)=\frac{\exp(\Gamma_{ij})}
{\sum_{l=1}^m\exp(\Gamma_{il})},
\end{equation}
so the model is exactly a Gumbel random-utility choice model with systematic utilities $\Gamma_{ij}$.

When the adjacent predictors are linear, $\gamma_{ij}=\bm{a}_{ij}'\bm{\beta}$, define
\[
\bm{z}_{ij}=\sum_{k=j}^{m-1}\bm{a}_{ik},\qquad j<m,
\qquad
\bm{z}_{im}=\bm{0}.
\]
Then $\Gamma_{ij}=\bm{z}_{ij}'\bm{\beta}$, the latent comparison matrix is the ordinary choice matrix $H_j$, and the structure vectors are
\begin{equation}
\label{eq:S-ACL}
\mathcal{S}_{\mathrm{ACL}}
=\{\bm{z}_{ij}-\bm{z}_{ik}:j\in[m],\ i\in I_j,\ k\in[m]_{-j}\}.
\end{equation}
For the common proportional form $\gamma_{ij}=\tilde{\alpha}_j-\bm{x}_i'\bm{\theta}$, with
$\bm{\beta}=(\tilde{\alpha}_1,\ldots,\tilde{\alpha}_{m-1},\bm{\theta}')'$, this gives
\[
\bm{z}_{ij}
=\left(\sum_{k=j}^{m-1}\bm{e}_k',-(m-j)\bm{x}_i'\right)',
\qquad j<m,
\qquad
\bm{z}_{im}=\bm{0}.
\]

\noindent\textbf{Fixed-score stereotype models.}
The ordered stereotype model with fixed scores \citep{polytomous:Anderson:1984} satisfies
\begin{equation}
\label{eq:stereotype-main}
\log\frac{\P(Y_i=j)}{\P(Y_i=m)}
=\alpha_j+\phi_j\bm{x}_i'\bm{\theta},
\qquad j=1,\ldots,m-1.
\end{equation}
This is a constrained baseline-category logit model. Set $\alpha_m=0$, $\bm{e}_m=\bm{0}$, and $\phi_m=0$. For fixed scores, the systematic utility of category $j$ is
\[
\eta_{ij}=\alpha_j+\phi_j\bm{x}_i'\bm{\theta}
=\bm{z}_{ij}'\bm{\beta},
\qquad
\bm{z}_{ij}=(\bm{e}_j',\phi_j\bm{x}_i')',
\qquad
\bm{\beta}=(\bm{\alpha}',\bm{\theta}')'.
\]
Hence
\[
\P(Y_i=j)=\frac{\exp(\eta_{ij})}
{\sum_{l=1}^m\exp(\eta_{il})},
\]
and the model is a constrained choice representation with ordinary choice comparisons. If the observed category is $y_i$, comparison with competitor $k$ gives the structure vector
\[
\bm{s}_{i,k}
=\bigl((\bm{e}_{y_i}-\bm{e}_k)',
((\phi_{y_i}-\phi_k)\bm{x}_i)'
\bigr)'.
\]
The corresponding structure set is
\begin{equation}
\label{eq:S-stereotype}
\mathcal{S}_{\mathrm{ST}}
=\{\bm{s}_{i,k}:i\in[n],\ k\in[m]_{-y_i}\}.
\end{equation}
For fixed scores the systematic utilities are linear in $(\bm{\alpha}',\bm{\theta}')'$, so this is an exact constrained choice representation with the displayed comparison vectors. When the stereotype scores are estimated jointly with the regression coefficients, the threshold map is bilinear and the linear structure-vector criterion is conditional on the scores.

\begin{table}[!tb]
\centering
\scriptsize
\renewcommand{\arraystretch}{1.08}
\setlength{\tabcolsep}{4pt}
\begin{tabular}{@{}>{\raggedright\arraybackslash}p{0.18\textwidth}>{\raggedright\arraybackslash}p{0.17\textwidth}>{\raggedright\arraybackslash}p{0.24\textwidth}>{\raggedright\arraybackslash}p{0.30\textwidth}@{}}
\toprule
Model & Latent event & Route into the family & Structure vectors \\
\midrule
Binary regression & half-line & direct & signed design vectors $\pm\bm{x}_i$ \\
Multinomial choice & utility comparisons & direct & $\bm{z}_{ij}-\bm{z}_{ik}$ \\
Cumulative link & interval & direct & upper- and lower-endpoint vectors \\
Rankings & utility comparisons & direct & $\bm{z}_{ia}-\bm{z}_{ib}$ \\
Multivariate binary & orthant & direct & signed component vectors \\
Multivariate ordinal & rectangle & direct & cutpoint and predictor blocks \\
Sequential & signed stage bounds & direct & stop and continuation vectors $\pm\bm{z}_{ik}$ \\
\addlinespace
Adjacent-category logit & utility comparisons & exact choice representation & effective-utility comparisons \\
Stereotype, fixed scores & utility comparisons & constrained choice representation & $\bm{s}_{i,k}$ \\
\bottomrule
\end{tabular}
\caption{Members of the latent polyhedral family and their route into the framework. After the stated representation, Theorem~\ref{thm:main} governs existence for every row. Section~\ref{sec:uniqueness} supplies common sufficient conditions for uniqueness. Sequential models use stagewise blocks or a pooled threshold map when coefficients are shared. Stereotype scores are fixed.}
\label{tab:family}
\end{table}

\subsection[Separation and the Estimable Span]{Separation and the Estimable Span}

Let $\mathcal{S}=\{\bm{s}_1,\ldots,\bm{s}_{\nS}\}$ be a non-empty finite set of structure vectors and let $V=\mathrm{span}(\mathcal{S})$. We say that $\mathcal{S}$ has \textbf{complete separation} if there is a nonzero $\bm{d}\in V$ such that $\bm{s}_i'\bm{d}>0$, $\forall i\in[\nS]$.
It has \textbf{quasi-complete separation} if there is a nonzero $\bm{d}\in V$ with $\bm{s}_i'\bm{d}\geq0$ for every $i$ and equality for at least one $i$. Finally, $\mathcal{S}$ has \textbf{overlap} if neither form of separation is present. Equivalently, overlap means that every nonzero $\bm{d}\in V$ satisfies $\bm{s}_i'\bm{d}<0$ for at least one structure vector. 

The definition has a direct likelihood interpretation. Along the parameter path $\bm{\beta}+t\bm{d}$, the threshold associated with $\bm{s}_i$ changes by $t\bm{s}_i'\bm{d}$. A separating direction therefore weakly enlarges every observed latent event and strictly enlarges at least one relevant threshold. Whether this prevents a finite maximizer depends only on how the corresponding probabilities respond to such threshold changes. This distributional condition is stated in Section~\ref{sec:unified-existence}.

Restricting $\bm{d}$ to $V$ removes directions that do not affect the represented likelihood. Every $\bm{\beta}\in\R^p$ has a decomposition $\bm{\beta}=\bar{\bm{\beta}}+\widehat{\bm{\beta}}$ with $\bar{\bm{\beta}}\in V$ and $\widehat{\bm{\beta}}\in V^\perp$. Since every structure vector lies in $V$, all thresholds satisfy
$\bm{b}_{ijr}'\bm{\beta}=\bm{b}_{ijr}'\bar{\bm{\beta}}$. The likelihood is therefore constant along $V^\perp$. Existence is determined on $V$, and uniqueness can hold only for the estimable component $\bar{\bm{\beta}}$. The next section develops the cone geometry needed to characterize escape on this subspace.

\section{Existence of Maximum Likelihood Estimates}
\label{sec:existence}

\subsection{Geometry of Overlap}
\label{sec:overlap-geometry}

Let $K=\mathrm{cone}(\mathcal{S})$ and $V=\mathrm{span}(\mathcal{S})$. The first lemma records equivalent geometric descriptions of overlap. The second strengthens them to the uniform negative bound that controls the likelihood at infinity. For a set $K\subseteq\R^p$, write $\mathrm{ri}(K)$ for its relative interior \citep[e.g.,][]{polytomous:Rockafellar:1990}.

\begin{lemma}
  \label{lemma:cone-A}
  Let $\mathcal{S} = \{ \bm{s}_1, \ldots, \bm{s}_{\nS} \}$ be a non-empty finite set of vectors
  $\bm{s}_i \in \R^p$, and $K = \mathrm{cone}(\mathcal{S})$ and $V = \mathrm{span}(\mathcal{S})$,
  respectively, be the convex cone and subspace generated by~$\mathcal{S}$. Then
  the following are equivalent:
  \begin{enumerate}[(a)]
   \item There is no $\bm{\beta}$ such that $\bm{s}_i'\bm{\beta} \ge 0$ for all
    $i \in [\nS]$ with strict inequality for at least one~$i$.
   \item $0 \in \mathrm{ri}(K)$.
   \item There exist positive $\lambda_1, \ldots, \lambda_{\nS}$ such
    that $\lambda_1 \bm{s}_1 + \cdots + \lambda_{\nS} \bm{s}_{\nS} = \bm{0}$.
   \item $K = V$.
   \item For all partitions of $\mathcal{S}$ into non-empty subsets $\mathcal{S}_1$ and
    $\mathcal{S}_2$, the intersection of $\mathrm{ri}(\mathrm{cone}(\mathcal{S}_1))$ and
    $\mathrm{ri}(\mathrm{cone}(-\mathcal{S}_2))$ is non-empty.
  \end{enumerate}
\end{lemma}

\begin{proof}
  Suppose first that $0 \not\in \mathrm{ri}(K)$.  Then using
  Theorem~11.2 in \citet{polytomous:Rockafellar:1990} with $M = \{0\}$,
  there exists a hyperplane $H$ containing $0$ (i.e., passing through
  the origin) such that one of the open half spaces associated with $H$
  contains $\mathrm{ri}(K)$.  Thus (possibly replacing $\bm{\beta}$ by
  $-\bm{\beta}$), there exists $\bm{\beta}$ (trivially non-zero) such that
  $\bm{z}'\bm{\beta} > 0$ for all $\bm{z} \in \mathrm{ri}(K)$.  As $K =
  \mathrm{cl}(\mathrm{ri}(K))$
  \citep[][Theorem~6.3]{polytomous:Rockafellar:1990}, we must have
  $\bm{s}_i'\bm{\beta} \ge 0$ for all $i$.  If none of these inequalities were
  strict, necessarily $\bm{z}'\bm{\beta} = 0$ for all $\bm{z} \in \mathrm{ri}(K)$, which
  is not possible.  Hence, not (b) implies not (a).

  Conversely, suppose there exists $\bm{\beta}$ such that $\bm{s}_i'\bm{\beta} \ge 0$
  for all $i$ with at least one inequality strict.  Clearly, $\bm{z}'\bm{\beta}
  \ge 0$ for all $\bm{z} \in K$.  If $0$ were in $\mathrm{ri}(K)$, then
  \citep[][Theorem~6.4]{polytomous:Rockafellar:1990} for all $\bm{z} \in K$
  there exists $\mu > 1$ such that $(1 - \mu) \bm{x} + \mu 0 = (1 - \mu) \bm{x}
  \in K$.  Take $\bm{z} = \bm{s}_i$ with $\bm{s}_i ' \bm{\beta} > 0$.  Then if $\mu > 1$,
  $((1 - \mu) \bm{s}_i)' \bm{\beta} < 0$, hence $(1 - \mu)\bm{s}_i \not\in K$ and we
  cannot have $0 \in \mathrm{ri}(K)$.  Hence, not (a) implies not (b),
  establishing the equivalence of (a) and (b).

  Equivalence of (b) and (c) is straightforward from the fact that
  $\mathrm{ri}(K)$ is the set of positive linear combinations of the
  $\bm{s}_i$.

  If (c) holds, then for all $\bm{z} = \sum_i \alpha_i \bm{s}_i$ in
  $\mathrm{span}(\mathcal{S})$ we have $\bm{z} = \sum_i (\alpha_i + \mu \lambda_i) \bm{s}_i$
  which is in $K$ for all $\mu$ sufficiently large, hence (d).
  Conversely, suppose that $K = \mathrm{span}(\mathcal{S})$.  Then for all $i$,
  $-\bm{s}_i \in \mathrm{span}(\mathcal{S}) = K$ so that there are non-negative
  $\gamma_{ij}$ such that $-\bm{s}_i = \sum_j \gamma_{ij} \bm{s}_j$, and hence
    \begin{displaymath}
    \bm{0}
    = \sum_i \bm{s}_i + \sum_i (-\bm{s}_i)
    = \sum_j \bm{s}_j + \sum_i \sum_j \gamma_{ij}\bm{s}_j
    = \sum_j \left(1 + \sum_i \gamma_{ij} \right) \bm{s}_j
    = \sum_j \lambda_j \bm{s}_j
  \end{displaymath}
    with all $\lambda_j > 0$.  Hence, (d) implies (c).

  Finally, consider a partition $\mathcal{S}_1$ and $\mathcal{S}_2$ of $\mathcal{S}$ into non-empty
  subsets, and write $I_1$ and $I_2$ for the corresponding index sets
  (i.e., $\mathcal{S}_j = \{ \bm{s}_i : i \in I_j \}$ with the $I_j$ non-empty with
  union $\{1, \ldots, \nS\}$).  The intersection of
  $\mathrm{ri}(\mathrm{cone}(\mathcal{S}_1))$ and
  $\mathrm{ri}(\mathrm{cone}(-\mathcal{S}_2))$ is non-empty iff there are positive
  $\lambda_i$ such that
  \begin{displaymath}
    \sum_{i \in I_1} \lambda_i \bm{s}_i = \sum_{i \in I_2} \lambda_i (-\bm{s}_i),
  \end{displaymath}
  or equivalently, $\sum_i \lambda_i \bm{s}_i = \bm{0}$.  Thus, (e) is equivalent
  to (c), and the proof is complete.
\end{proof}

\begin{lemma}
\label{lemma:cone-B}
Let $\mathcal{S} = \{ \bm{s}_1, \ldots, \bm{s}_{\nS} \}$ be a non-empty finite set of vectors
$\bm{s}_i \in \R^p$, let $K = \mathrm{cone}(\mathcal{S})$ and $V=\mathrm{span}(\mathcal{S})$, and define
$\varphi_{\mathcal{S}}(\bm{\beta}) = \min_{i \in [\nS]} \bm{s}_i'\bm{\beta}$. Then the following are equivalent:
\begin{enumerate}[(a)]
\item $K = V$.
\item There exists $C < 0$ such that, for all $\bm{\beta} \in V$,
\[
\varphi_{\mathcal{S}}(\bm{\beta}) \leq C \| \bm{\beta} \|.
\]
\item For every $\bm{\beta} \in V \setminus \{\bm{0}\}$, there is an $i \in [\nS]$ with $\bm{s}_i'\bm{\beta}<0$.
\end{enumerate}
\end{lemma}

\begin{proof}
If $V=\{0\}$, all three conditions hold trivially. Hence assume that $\dim(V)\geq 1$. On the subspace $V=\mathrm{span}(\mathcal{S})$, condition (a) is equivalent to the absence of $\bm{\beta}\in V\setminus\{\bm{0}\}$ with $\bm{s}_i'\bm{\beta}\ge0$ for all $i$, by Lemma~1.
For $\bm{\beta}\ne\bm{0}$, $\varphi_{\mathcal{S}}(\bm{\beta})/\|\bm{\beta}\|
=\varphi_{\mathcal{S}}(\bm{\beta}/\|\bm{\beta}\|)$, so (b) holds iff
$\max_{\bm{\beta}\in V:\|\bm{\beta}\|=1}\varphi_{\mathcal{S}}(\bm{\beta})<0$
(the maximum exists by compactness and continuity), which is equivalent to (c).
\end{proof}

Condition (c) in Lemma~\ref{lemma:cone-B} is the definition of overlap on $V$. Its uniform version in part (b) is the key analytic consequence. Every parameter vector far from the origin makes at least one observed threshold decrease at a rate proportional to its norm.

\subsection{Existence under Overlap}
\label{sec:unified-existence}

For the converse direction we need observed event probabilities to react strictly to effective threshold increases. A CDF $G$ on $\R^q$ is \emph{strictly threshold-increasing on its positive region} if
\begin{equation}
\label{eq:strict-threshold}
G(\bm{u})>0,\quad \bm{a}\in[0,\infty)^q\setminus\{\bm{0}\}
\quad\Longrightarrow\quad
G(\bm{u}+\bm{a})>G(\bm{u}).
\end{equation}

\begin{theorem}
  \label{thm:main}
  Consider the categorical response likelihood in Equation~\ref{eq:likelihood}, regarded as a function on $\R^p$. Suppose that every observed-category factor has the representation
  \[
  h_j(Z_i'\bm{\beta})=G_{ij}(B_{ij}'\bm{\beta}),\qquad i\in I_j,
  \]
  with pooled structure set $\mathcal{S}$ from Equation~\ref{eq:S}. If all $G_{ij}$ are continuous and $\mathcal{S}$ has overlap, then a finite MLE exists.

  If all $G_{ij}$ are continuous, the likelihood is positive at some finite parameter value,
  and every $G_{ij}$ is strictly threshold-increasing on its positive region 
  (Equation~\ref{eq:strict-threshold}), then a finite MLE exists if and only if $\mathcal{S}$ has overlap.
\end{theorem}

\begin{proof}
Let $V=\mathrm{span}(\mathcal{S})$. For necessity, suppose that overlap fails. Then some $\bm{d}\in V\setminus\{\bm{0}\}$ satisfies $\bm{b}_{ijr}'\bm{d}\geq0$ for all observed $(i,j,r)$, with at least one strict inequality. Otherwise $\bm{d}$ would be orthogonal to its own spanning set. If $\bm{\beta}_0$ were a finite maximizer, positivity of the likelihood somewhere would imply $L(\bm{\beta}_0)>0$. With $\bm{a}_{ij}=B_{ij}'\bm{d}\geq\bm{0}$, coordinatewise monotonicity gives
\[
G_{ij}(B_{ij}'\bm{\beta}_0+t\bm{a}_{ij})
\geq G_{ij}(B_{ij}'\bm{\beta}_0),\qquad t>0,
\]
for every observed pair, and strict threshold monotonicity makes at least one inequality strict. Hence $L(\bm{\beta}_0+t\bm{d})>L(\bm{\beta}_0)$, a contradiction.

For sufficiency, assume overlap. Since the likelihood depends only on the inner products with vectors in $\mathcal{S}$,
$L(\bm{\beta})=L(\mathrm{Proj}_V\bm{\beta})$ and the supremum over $\R^p$ equals that over $V$. The identically zero case is trivial, so take a point with positive likelihood. For $\bm{b}=\bm{b}_{ijr}$, let $F_{\bm{b}}$ be the marginal CDF of the $r$th coordinate of $\bm{W}_{ij}$. Because all likelihood factors lie in $[0,1]$ and $y_{ij}\geq1$,
\[
L(\bm{\beta})\leq G_{ij}(B_{ij}'\bm{\beta})
\leq F_{\bm{b}}(\bm{b}'\bm{\beta}).
\]
Choose $\bm{b}^{*}(\bm{\beta})$ minimizing $\bm{b}'\bm{\beta}$ over $\mathcal{S}$. Lemma~2 gives $C<0$ with $(\bm{b}^{*}(\bm{\beta}))'\bm{\beta}\leq C\|\bm{\beta}\|$ on $V$. If $F^{(1)},\ldots,F^{(Q)}$ are the finitely many corresponding marginal CDFs, then
\[
L(\bm{\beta})\leq\max_{s=1,\ldots,Q}F^{(s)}(C\|\bm{\beta}\|)\longrightarrow0
\quad\text{as }\|\bm{\beta}\|\to\infty,\ \bm{\beta}\in V.
\]
Thus the supremum is attained on a sufficiently large compact ball in $V$, because the likelihood is continuous. This proves existence of a finite maximizer.
\end{proof}

The theorem separates geometry from probability. Under overlap, Lemma~\ref{lemma:cone-B} forces at least one threshold to tend to $-\infty$ along every sequence that diverges in $V=\mathrm{span}(\mathcal{S})$. The associated observed probability tends to zero, so the full likelihood tends to zero. Continuity then gives a maximizer on a compact subset of $V$. Conversely, a separating direction weakly raises every observed threshold and strictly raises at least one. Strict threshold monotonicity then improves the likelihood from every finite point. The argument requires neither likelihood concavity nor exchangeability.

If $V\subsetneq\R^p$, the likelihood is constant on the affine spaces $\bm{\beta}+V^\perp$. A maximizer on $V$ therefore generates the ambient maximizer set by adding arbitrary vectors in $V^\perp$. When the likelihood is positive somewhere, overlap also makes the maximizer set bounded on $V$. This formulation separates nonidentifiability from nonexistence.

\subsection{Model Consequences}
\label{sec:model-consequences}
\label{sec:extensions}

The model representations are already fixed in Section~\ref{sec:framework}; applying Theorem~\ref{thm:main} requires only the corresponding continuity and strict monotonicity conditions.

\noindent\textbf{Choice-type and binary models.}
For the choice representation in Section~\ref{sec:choice-models}, Theorem~\ref{thm:main} applies with $\mathcal{S}=\mathcal{S}_{\mathrm{CT}}$. Continuity of every utility-difference CDF $G_j$ makes overlap of $\mathcal{S}_{\mathrm{CT}}$ sufficient for a finite MLE. If each utility-difference distribution has a density that is positive on $\R^{m-1}$, every effective comparison threshold is strictly probability increasing, and existence is equivalent to overlap. In the binary case this recovers the classical criterion of \citet{polytomous:Albert+Anderson:1984} for the signed design vectors in Section~\ref{sec:choice-models}. Sufficiency continues to require only continuity and therefore includes heavy-tailed links such as cauchit. The support qualification is substantive. With one-sided or bounded support (e.g., with the log link), a separating path may eventually move a threshold beyond the support, after which the corresponding probability is constant. Separation can then coexist with a finite maximizer, although the maximizer set may be unbounded.

\noindent\textbf{Cumulative-link models.}
For the cumulative representation in Section~\ref{sec:cumulative-models}, use $\mathcal{S}=\mathcal{S}_{\mathrm{CL}}$. For an intermediate category,
$G_{ij}(u_1,u_2)=\P(-u_2\leq\epsilon_i\leq u_1)$. This is continuous when $F$ is continuous, and a density positive on $\R$ implies strict threshold monotonicity whenever the interval has positive probability. The boundary categories are the one-dimensional cases.

To recover the ordered parameter space, suppose $\mathcal{S}_{\mathrm{CL}}$ has overlap. If some response category is empty, the nonempty and empty categories must meet at some boundary $r\in[m-1]$. If $I_r\neq\varnothing$ and $I_{r+1}=\varnothing$, let $\bm{d}_0$ have threshold component $\bm{e}_r$ and zero regression component. If $I_r=\varnothing$ and $I_{r+1}\neq\varnothing$, use threshold component $-\bm{e}_r$ instead. In either case, every $\bm{s}\in\mathcal{S}_{\mathrm{CL}}$ satisfies $\bm{s}'\bm{d}_0\geq0$, and at least one inequality is strict. With $V=\mathrm{span}(\mathcal{S}_{\mathrm{CL}})$ and $\bm{d}=\mathrm{Proj}_V\bm{d}_0$, we have $\bm{s}'\bm{d}=\bm{s}'\bm{d}_0$ for all $\bm{s}\in\mathcal{S}_{\mathrm{CL}}$. The strict inequality also implies $\bm{d}\neq\bm{0}$. This contradicts overlap. Hence $I_j\neq\varnothing$ for every $j\in[m]$.

Define the interval CDF for arbitrary thresholds by assigning probability zero when the lower endpoint is not below the upper endpoint. The resulting extended likelihood is continuous on $\R^p$. If its maximum is positive, every category probability at a maximizer is positive, and the presence of every category implies $\alpha_1<\cdots<\alpha_{m-1}$. Thus the maximizer already lies in the original ordered parameter space. If the extended likelihood is identically zero, any finite ordered parameter is a maximizer. For necessity, a separating direction improves the likelihood locally. Because the ordered parameter space is open, $\bm{\beta}+t\bm{d}$ remains ordered for all sufficiently small $t>0$. Therefore Theorem~1 gives the characterization for cumulative logit, ordered probit, and cumulative cloglog without the likelihood-concavity and full-rank assumptions used by \citet{polytomous:Tse:1986}.

\noindent\textbf{Other direct members.}
The ranking, multivariate binary, multivariate ordinal, and sequential representations in Section~\ref{sec:direct-members} enter Theorem~\ref{thm:main} without another reduction. Continuity of the relevant comparison, orthant, rectangle, or stage-event probabilities gives sufficiency of overlap. A positive joint latent density gives strict threshold monotonicity and hence necessity for the full-row-rank and coordinate-bound representations used here. A partially observed ranking requires an effective comparison system, obtained by removing probability-redundant inequalities or by verifying strict threshold monotonicity directly. The same ordering argument applies coordinatewise to multivariate ordinal models. If one marginal category of a response coordinate is empty, changing the adjacent cutpoint in the favorable direction gives a separating vector for every observed joint pattern. Under overlap, every marginal category is therefore represented. A positive joint likelihood then forces the lower and upper cutpoints selected by every observed coordinate category to be strictly ordered. Thus the extended rectangle likelihood and the original ordered-cutpoint likelihood have the same maximizers. For sequential models with stage-specific parameter blocks, the pooled structure set is a direct sum of the stagewise signed vectors, so overlap is equivalent to overlap at every stage. When coefficients are shared, all stop and continuation vectors are pooled in the common parameter space.

Dependence among latent components changes the event probabilities but not the structure vectors, so it does not alter the existence criterion when the required continuity and positivity conditions hold.

\noindent\textbf{Ordinal models with choice representations.}
The adjacent-category logit and fixed-score stereotype models in
Section~\ref{sec:choice-ordinal} use exact choice representations and
therefore inherit the choice conclusion with their own effective
comparison vectors. For adjacent-category logit these are differences
of the effective utilities $\Gamma_{ij}$. For the fixed-score
stereotype model they are the constrained vectors $\bm{s}_{i,k}$.
No encompassing unconstrained model is needed for either exact
characterization.

When the stereotype scores are estimated jointly with the regression
coefficients, the threshold map is bilinear and the model falls outside
the present linear-threshold framework. Nevertheless, it is a
constrained submodel of the encompassing multinomial logit model.
Consequently, overlap in the encompassing model is sufficient for the
existence of a finite stereotype-model maximizer, provided the
normalized stereotype parameter space is closed (e.g., the usual \(
0=\phi_1\leq\phi_2\leq\cdots\leq\phi_m=1
\)).

\section[Detection and Geometry of Separation]{Detection and Geometry of Separation}
\label{sec:diagnostics}

\subsection{Linear-Programming Detection}
\label{comp-detect}

Let $S$ be the matrix whose rows are the pooled structure vectors in $\mathcal{S}$. Separation can be detected by solving the ordinary primal problem
\begin{gather}
\label{lp1}
\text{maximize}\quad \bm{1}'\bm{t} \\
\text{subject to}\quad S\bm{\beta}-\bm{t}=\bm{0},\qquad
\bm{t}\geq\bm{0}.\nonumber
\end{gather}
A feasible point has $\bm{t}=S\bm{\beta}$. Thus the linear program is unbounded if and only if there is a separating direction, that is, some $\bm{\beta}$ with $S\bm{\beta}\geq\bm{0}$ and at least one strict inequality. Multiplying such a direction by arbitrarily large positive constants makes the objective diverge. Under overlap, no such direction exists and the only feasible value of $\bm{t}$ is $\bm{0}$.

Through duality theory of linear programming \citep[e.g.,][Ch.~7.4]{polytomous:Schrijver:1986}, after a standard sign change of the dual variables, the corresponding dual feasibility problem is
\begin{gather*}
\text{minimize}\quad \bm{0}'\bm{u} \\
\text{subject to}\quad S'\bm{u}=\bm{0},\qquad
\bm{u}\geq\bm{1}.
\end{gather*}
The unboundedness of the primal problem is equivalent to infeasibility of this dual problem. Hence overlap is equivalent to the existence of a vector $\bm{u}\geq\bm{1}$ satisfying $S'\bm{u}=\bm{0}$. Because any strictly positive null vector can be rescaled to satisfy $\bm{u}\geq\bm{1}$, Lemma~\ref{lemma:cone-A}(c) shows that this is exactly the positive-combination characterization of overlap.

For numerical use we solve the equivalent, always-feasible formulation
\begin{gather}
\label{lp2}
\text{minimize}\quad y \\
\text{subject to}\quad S'\bm{u}=\bm{0},\qquad
\bm{u}+y\bm{1}\geq\bm{1},\qquad y\geq0.\nonumber
\end{gather}
Overlap holds if and only if the optimal value is zero. At $y=0$, Equation~\ref{lp2} is exactly the dual feasibility condition above; if no such strictly positive null vector exists, the slack $y$ records the violation.

The formulation works with the structure-vector matrix rather than with a model-specific design matrix. It can therefore be used across the model classes introduced in Section~\ref{sec:framework}. Related linear programs for log-linear count models were developed by \citet{Fienberg+Rinaldo:2012} using the marginal cone. Variants of Equation~\ref{lp2} can also localize observations and model terms associated with separation, as illustrated in Section~\ref{sec:application} and shown in detail by \citet{sep-JSS}.

\subsection{The Cone of Separating Directions}
\label{sec:dim}

Let $K=\mathrm{cone}(\mathcal{S})$ and let $L=K\cap(-K)$ be its lineality space. The lineality space is the largest subspace contained in $K$. It records directions represented by the cone in both signs. The dimension of the lineality space determines the dimension of the dual cone of directions along which all observed thresholds can move favorably.

\begin{lemma}
\label{lemma:dimension}
Let $\mathcal{S}$ be a finite non-empty set of vectors in $\R^p$, let
$V=\mathrm{span}(\mathcal{S})$, and let $K=\mathrm{cone}(\mathcal{S})\subseteq V$. Write $K=L+P$, where $L=K\cap(-K)$ is the lineality space, $P$ is pointed, and $L\cap P=\{\bm{0}\}$. Then:
\begin{enumerate}[(a)]
\item Under overlap, $L=V$.
\item Under complete separation, $L=\{\bm{0}\}$.
\item Under quasi-complete separation, the cone of separating directions
\[
R=\{\bm{\gamma}\in V:\bm{s}'\bm{\gamma}\geq0
\text{ for every }\bm{s}\in\mathcal{S}\}
\]
has dimension $\dim(R)=\dim(K)-\dim(L)$.
\end{enumerate}
\end{lemma}

\begin{proof}
  For (a), overlap on $V$ implies $K = V$ by Lemma~2, hence $L = V$.
   For (b), complete separation means there exists a $\bm{\beta} \in V$ such that $\bm{s}'\bm{\beta} > 0$ for all $\bm{s} \in \mathcal{S}$. This implies $K$ lies strictly within an open half-space of $V$. If a non-zero $\bm{v}$ were in $L$, then both $\bm{v}$ and $-\bm{v}$ are in $K$, leading to the contradiction $\bm{v}'\bm{\beta} > 0$ and $(-\bm{v})'\bm{\beta} > 0$. Thus, $L=\{\bm{0}\}$.
   For (c), \(R\) is the dual cone of \(K\) in \(V\). Its span is
   \(L^\perp\cap V\), since the lineality space \(L=K\cap(-K)\) is exactly the subspace annihilated by all separating directions. Hence
   \[
   \dim(R)=\dim(V)-\dim(L)=\dim(K)-\dim(L),
   \] because \(K\) spans \(V\).
\end{proof}

The dimension of $R$ is the number of linearly independent separating directions. It describes the multiplicity of separation, not its geometric margin or numerical severity. Complete separation has zero lineality, while quasi-complete separation can retain a nontrivial lineality space.

A useful monotonicity observation follows directly. If a subset of the data has overlap and its structure vectors span the same subspace as the full sample, then the full sample also has overlap. This can provide an early certificate in sequential or distributed screening. Separation in a subset is not conclusive because additional observations can restore overlap.

\section[Strict Log-Concavity and Uniqueness]{Strict Log-Concavity and Uniqueness}
\label{sec:uniqueness}

Having established conditions for existence, we now turn to uniqueness. The latent polyhedral representation makes it possible to derive strict log-concavity from the primitive latent distribution. We first recall a result for strictly log-concave probability measures. We then apply it to parameterized polyhedra and finally transfer the resulting strictness to the identifiable parameter space.

\subsection{Strict Log-Concavity of Latent Polyhedral Probabilities}
\label{sec:logcmeasures}

Let $D$ be a convex subset of $\R^q$. A function $g:D\to\R_+$ is strictly log-concave on $D$ if, for distinct $\bm{u},\bm{v}\in D$ and every $0<\lambda<1$,
\begin{equation}
 g(\lambda\bm{u}+(1-\lambda)\bm{v})
 >
 g(\bm{u})^\lambda g(\bm{v})^{1-\lambda}.
\end{equation}
The corresponding setwise property is supplied by the following result. For sets $A$ and $B$, write
$A\mathbin{\triangle}B=(A\setminus B)\cup(B\setminus A)$ for their symmetric difference.

\citet[Theorem~1]{Aristodemou2023} proves that a positive and
strictly log-concave density induces a strictly log-concave probability
measure. We record the result in the probability-distinct form used
below. The condition
$\P(A\mathbin{\triangle}B)>0$ in particular implies that $A\neq B$.

\begin{theorem}
\label{thm:aristodemou}
Let $f$ be a positive and strictly log-concave density on an open
convex set $D\subset\R^r$, and let $\P$ be its probability measure.
For convex sets $A,B\subseteq D$ with $\P(A)>0$, $\P(B)>0$, and
$\P(A\mathbin{\triangle}B)>0$, every $0<\lambda<1$ satisfies
\begin{equation}
 \P\!\left(\lambda A+(1-\lambda)B\right)
 >
 \P(A)^\lambda\P(B)^{1-\lambda}.
\end{equation}
\end{theorem}

For a fixed matrix $H\in\R^{q\times r}$ and a latent vector $\bm{U}$ supported on $D$, define
\begin{equation}
\label{eq:threshold-polyhedron}
 \mathcal{C}_{H}(\bm{b})
 =\{\bm{u}\in D:H\bm{u}\leq\bm{b}\},
 \qquad
 g_H(\bm{b})=\P\{\bm{U}\in\mathcal{C}_{H}(\bm{b})\}.
\end{equation}
Write $\mathcal{B}_{H}^{+}=\{\bm{b}:g_H(\bm{b})>0\}$ for the positive threshold region.

\begin{definition}[Threshold-identifiable representation]
\label{def:threshold-identifiable}
The representation in Equation~\ref{eq:threshold-polyhedron} is threshold-identifiable on $\mathcal{B}_{H}^{+}$ if, for all distinct $\bm{b}_0,\bm{b}_1\in\mathcal{B}_{H}^{+}$,
\begin{equation}
 \P\{\mathcal{C}_{H}(\bm{b}_0)\mathbin{\triangle}
       \mathcal{C}_{H}(\bm{b}_1)\}>0.
\end{equation}
\end{definition}

Threshold identifiability rules out threshold coordinates that do not change the latent event on its positive region. It is weaker than full row rank of $H$ and therefore also permits singular transformed latent vectors.

\begin{theorem}
\label{thm:polyhedral-logconcavity}
Let $\bm{U}$ have a positive and strictly log-concave density on an open convex set $D\subset\R^r$. If the representation in Equation~\ref{eq:threshold-polyhedron} is threshold-identifiable on $\mathcal{B}_{H}^{+}$, then $\mathcal{B}_{H}^{+}$ is convex and $g_H$ is strictly log-concave on $\mathcal{B}_{H}^{+}$.
\end{theorem}

\begin{proof}
Take distinct
$\bm b_0,\bm b_1\in\mathcal B_H^{+}$
and let $0<\lambda<1$. Define
\[
\bm b_\lambda
=
\lambda\bm b_0+(1-\lambda)\bm b_1.
\]

We first compare the corresponding latent polyhedra. Let
\[
\bm u_0\in\mathcal C_H(\bm b_0),
\qquad
\bm u_1\in\mathcal C_H(\bm b_1).
\]
By definition of $\mathcal C_H(\bm b)$, this means
\[
\bm u_0,\bm u_1\in D,
\qquad
H\bm u_0\leq\bm b_0,
\qquad
H\bm u_1\leq\bm b_1.
\]
Since $D$ is convex,
\[
\bm u_\lambda
=
\lambda\bm u_0+(1-\lambda)\bm u_1
\in D.
\]
Moreover, linearity of $H$ gives
\[
H\bm u_\lambda
=
H\{\lambda\bm u_0+(1-\lambda)\bm u_1\}
=
\lambda H\bm u_0+(1-\lambda)H\bm u_1.
\]
Because $0<\lambda<1$, multiplication of the two componentwise
inequalities above by $\lambda$ and $1-\lambda$, respectively,
preserves their directions. Hence
\[
\lambda H\bm u_0+(1-\lambda)H\bm u_1
\leq
\lambda\bm b_0+(1-\lambda)\bm b_1
=
\bm b_\lambda.
\]
Thus
\[
\bm u_\lambda\in\mathcal C_H(\bm b_\lambda).
\]

Since $\bm u_0$ and $\bm u_1$ were arbitrary, every element of the
Minkowski combination
\[
\lambda\mathcal C_H(\bm b_0)
+
(1-\lambda)\mathcal C_H(\bm b_1)
\]
belongs to $\mathcal C_H(\bm b_\lambda)$. Therefore,
\begin{equation} \label{S1}
    \lambda\mathcal C_H(\bm b_0)
+
(1-\lambda)\mathcal C_H(\bm b_1)
\subseteq
\mathcal C_H(\bm b_\lambda).
\end{equation}

Next, because
$\bm b_0,\bm b_1\in\mathcal B_H^{+}$,
both endpoint events have positive probability:
\[
\P\{\mathcal C_H(\bm b_0)\}>0,
\qquad
\P\{\mathcal C_H(\bm b_1)\}>0.
\]
Since $\bm b_0\neq\bm b_1$, threshold identifiability implies
\[
\P\!\left\{
\mathcal C_H(\bm b_0)
\mathbin{\triangle}
\mathcal C_H(\bm b_1)
\right\}>0.
\]
The sets $\mathcal C_H(\bm b_0)$ and
$\mathcal C_H(\bm b_1)$ are convex subsets of $D$, so
Theorem~2 applies and yields
\begin{equation} \label{S2}
    \P\!\left\{
\lambda\mathcal C_H(\bm b_0)
+
(1-\lambda)\mathcal C_H(\bm b_1)
\right\}
>
\P\{\mathcal C_H(\bm b_0)\}^{\lambda}
\P\{\mathcal C_H(\bm b_1)\}^{1-\lambda}.
\end{equation}

Finally, probability is monotone under set inclusion. Applying this to
\eqref{S1} gives
\[
\P\{\mathcal C_H(\bm b_\lambda)\}
\geq
\P\!\left\{
\lambda\mathcal C_H(\bm b_0)
+
(1-\lambda)\mathcal C_H(\bm b_1)
\right\}.
\]
Combining this inequality with \eqref{S2}, and recalling that
\[
g_H(\bm b)
=
\P\{\mathcal C_H(\bm b)\},
\]
we obtain
\[
g_H(\bm b_\lambda)
>
g_H(\bm b_0)^\lambda
g_H(\bm b_1)^{1-\lambda}
>
0.
\]
In particular, $\bm b_\lambda\in\mathcal B_H^{+}$. Since this holds
for every $\bm b_0,\bm b_1\in\mathcal B_H^{+}$ and
$0<\lambda<1$, the set $\mathcal B_H^{+}$ is convex. The displayed
strict inequality is precisely strict log-concavity of $g_H$ on
$\mathcal B_H^{+}$.
\end{proof}

The proof works in the original latent space. For two threshold vectors, the Minkowski combination of their polyhedra is contained in the polyhedron at the combined threshold. Theorem~\ref{thm:aristodemou} makes the resulting probability inequality strict. 

\begin{corollary}
\label{cor:surjective-thresholds}
Let $\bm{U}$ have a positive and strictly log-concave density on $\R^r$. If $H\in\R^{q\times r}$ has full row rank, then the representation is threshold-identifiable on $\R^q$ and
\begin{equation}
 \bm{b}\longmapsto\P(H\bm{U}\leq\bm{b})
\end{equation}
is strictly log-concave on $\R^q$.
\end{corollary}

\begin{proof}

For $\bm{b}\in\R^q$, let
\[
\mathcal{L}(\bm{b})
=
(-\infty,b_1]\times\cdots\times(-\infty,b_q].
\]
Then
\[
\mathcal{C}_{H}(\bm{b})
=
H^{-1}\{\mathcal{L}(\bm{b})\},
\]
where $H^{-1}$ denotes the set-theoretic inverse image.

Take $\bm{b}_0\neq\bm{b}_1$. Then the two lower orthants differ in
at least one coordinate. Hence their symmetric difference contains a
nonempty open set $O\subset\R^q$. Since $H$ has full row rank, the
linear map $H:\R^r\to\R^q$ is surjective. Therefore every point of
$O$ has a preimage under $H$, so $H^{-1}(O)$ is nonempty. Moreover,
$H^{-1}(O)$ is open because $H$ is continuous. Since
\[
O
\subseteq
\mathcal{L}(\bm{b}_0)\mathbin{\triangle}\mathcal{L}(\bm{b}_1),
\]
taking inverse images gives
\[
H^{-1}(O)
\subseteq
\mathcal{C}_{H}(\bm{b}_0)
\mathbin{\triangle}
\mathcal{C}_{H}(\bm{b}_1).
\]

The primitive density is positive on $\R^r$, so every nonempty open
subset of $\R^r$ has positive probability. Hence
\[
\P\{H^{-1}(O)\}>0,
\]
and consequently
\[
\P\!\left\{
\mathcal{C}_{H}(\bm{b}_0)
\mathbin{\triangle}
\mathcal{C}_{H}(\bm{b}_1)
\right\}>0.
\]
Thus the representation is threshold-identifiable.

Finally, for any $\bm b\in\R^q$, the open lower orthant
\[
(-\infty,b_1)\times\cdots\times(-\infty,b_q)
\]
is nonempty. By the same surjectivity and continuity argument, its
inverse image under $H$ is a nonempty open subset of
$\mathcal C_H(\bm b)$ and therefore has positive probability. Hence
$g_H(\bm b)>0$ for every $\bm b\in\R^q$, so
$\mathcal B_H^+=\R^q$. Theorem~3 now gives the result.

\end{proof}

Full row rank makes $H$ surjective onto $\R^q$, so every vector of comparison values can occur as $H\bm{u}$ for some latent state. No threshold coordinate is inactive, and distinct threshold vectors define genuinely different events.

For choice models, $H_j$ has rank $m-1$. The only lost direction is a common shift of all utilities, which cannot affect a choice. Corollary~\ref{cor:surjective-thresholds} therefore gives strict log-concavity of the utility-difference probabilities whenever the primitive utility-error vector has a positive strictly log-concave joint density. This covers multinomial logit under independent Gumbel errors and multinomial probit under Gaussian errors, while also allowing dependent and nonexchangeable errors. The adjacent-category logit and fixed-score stereotype models use the
same choice matrix, so the conclusion applies on the spans of their
effective comparison vectors. When the stereotype scores are estimated
jointly with the regression coefficients, the parameter space becomes
a generally nonconvex rank-one constrained subset of the encompassing
multinomial logit model, so strict concavity of the encompassing
likelihood does not by itself imply uniqueness for the stereotype
estimator.

Complete rankings use a path-incidence comparison matrix of rank $J-1$, and multivariate binary models use a nonsingular diagonal sign matrix. Sequential category events use signed coordinate-selection matrices. For $j<m$, the matrix in Equation~\ref{eq:sequential-polyhedral} satisfies $H_j^{\mathrm{seq}}(H_j^{\mathrm{seq}})'=I_j$, while $H_m^{\mathrm{seq}}=-I_{m-1}$. These matrices therefore have full row rank. If the independent stage-error densities are positive and strictly log-concave, their product is a positive strictly log-concave joint density, and the corollary covers sequential logit and probit models as well. Partial rankings are covered whenever their selected comparison system has full row rank. Otherwise Theorem~\ref{thm:polyhedral-logconcavity} can be applied directly after redundant comparisons are removed or threshold identifiability is verified.

\begin{corollary}
\label{cor:coordinate-bounds}
Let $\bm{U}\in\R^d$ have a positive and strictly log-concave density. For $K\subseteq[d]$, let $E_K$ have rows $\bm{e}_k'$, $k\in K$. For index sets $K_u,K_l\subseteq[d]$, not both empty, define
\begin{equation}
\begin{gathered}
 H_{K_l,K_u}=\begin{pmatrix}E_{K_u}\\-E_{K_l}\end{pmatrix},
 \qquad
 \bm{b}_{K_l,K_u}=\begin{pmatrix}\bm{u}_{K_u}\\-\bm{l}_{K_l}\end{pmatrix},\\
 \mathcal{B}_{K_l,K_u}^{+}
 =\{\bm{b}_{K_l,K_u}:l_k<u_k\text{ for every }k\in K_l\cap K_u\}.
\end{gathered}
\end{equation}
The representation through $H_{K_l,K_u}$ is threshold-identifiable on $\mathcal{B}_{K_l,K_u}^{+}$. Consequently, $\bm{b}\mapsto\P(H_{K_l,K_u}\bm{U}\leq\bm{b})$ is strictly log-concave there.
\end{corollary}

\begin{proof}
For $\bm{b}=\bm{b}_{K_l,K_u}$,
\[
\mathcal{C}_{H_{K_l,K_u}}(\bm{b})
=
\{\bm{x}\in\R^d:l_k\leq x_k\text{ for }k\in K_l,\ x_k\leq u_k\text{ for }k\in K_u\}.
\]
This set has positive probability exactly when $l_k<u_k$ for every $k\in K_l\cap K_u$. Hence its positive threshold region is $\mathcal{B}_{K_l,K_u}^{+}$, which is convex.

Take distinct $\bm{b}_0,\bm{b}_1\in\mathcal{B}_{K_l,K_u}^{+}$. At least one selected lower or upper bound differs. The two corresponding coordinate-bound polyhedra then have a symmetric difference containing a nonempty open rectangle. Positivity of the density gives this symmetric difference positive probability. The representation is therefore threshold-identifiable, and Theorem~3 gives strict log-concavity on $\mathcal{B}_{K_l,K_u}^{+}$.
\end{proof}

Coordinate-bound matrices are generally not covered by the full row rank corollary. If a coordinate has both a lower and an upper bound, the matrix contains the dependent rows $\bm{e}_k'$ and $-\bm{e}_k'$ and cannot be surjective onto the full threshold space. It is nevertheless threshold-identifiable. Each admissible bound determines an actual face of the polyhedron, and changing that bound adds or removes a region with positive volume and therefore positive probability.

A cumulative-link factor is the one-dimensional case. Interior categories use both a lower and an upper bound, while the first and last categories retain only one of them. Hence positive strictly log-concave scalar error densities give strictly log-concave category probabilities for cumulative logit, probit, and complementary log-log models. Multivariate ordinal factors use the same construction coordinatewise, producing rectangles with one-sided bounds at boundary categories. The corollary therefore also covers bivariate and multivariate ordered probit with fixed positive definite covariance, and more generally any multivariate ordinal model whose latent density is positive and strictly log-concave.

\subsection{Uniqueness on the Estimable Span}

\begin{theorem}
\label{thm:latent-polyhedral-uniqueness}
Consider a latent polyhedral categorical response likelihood with pooled structure set $\mathcal{S}$ and let $V=\mathrm{span}(\mathcal{S})$. Suppose that the likelihood is positive at some finite parameter value. For every observed factor $(i,j)$, assume that the primitive latent vector $\bm{U}_i$ has a positive and strictly log-concave density on an open convex support and that its representation through $H_{ij}$ is threshold-identifiable on its positive threshold region.

Then a finite MLE exists if and only if $\mathcal{S}$ has overlap. When overlap holds, the log-likelihood is strictly concave on
\begin{equation}
 \mathcal{D}=\{\bm{\beta}\in V:L(\bm{\beta})>0\},
\end{equation}
and the MLE is unique on $V$. Its full set of maximizers in $\R^p$ is
\begin{equation}
 \widehat{\bm{\beta}}+V^\perp.
\end{equation}
If the log-likelihood is differentiable on $\mathcal{D}$, it has no other critical point there.
\end{theorem}

\begin{proof}
For every observed factor $(i,j)$, Theorem~3 shows that
\[
G_{ij}(\bm{b})
=
\P(H_{ij}\bm{U}_i\leq\bm{b})
\]
is strictly log-concave on its positive threshold region. These regions are convex. Their linear inverse images therefore have a convex intersection, so
\[
\mathcal{D}=\{\bm{\beta}\in V:L(\bm{\beta})>0\}
\]
is convex.

Take distinct $\bm{\beta}_0,\bm{\beta}_1\in\mathcal{D}$ and set $\bm{d}=\bm{\beta}_0-\bm{\beta}_1$. Since $\bm{d}\in V\setminus\{\bm{0}\}$ and $V=\operatorname{span}(\mathcal{S})$, some structure vector $\bm{b}_{ijr}\in\mathcal{S}$ satisfies $\bm{b}_{ijr}'\bm{d}\neq0$. Hence
\[
B_{ij}'\bm{\beta}_0\neq B_{ij}'\bm{\beta}_1
\]
for at least one observed factor. Strict log-concavity gives a strict inequality for that factor at every nontrivial convex combination of $\bm{\beta}_0$ and $\bm{\beta}_1$. The remaining factors give weak inequalities. Multiplication by $y_{ij}>0$ and summation show that $\ell$ is strictly concave on $\mathcal{D}$.

The induced CDFs are continuous because the primitive latent vectors have densities. Threshold identifiability also gives strict threshold monotonicity. Indeed, if $\bm{b}_1\geq\bm{b}_0$ coordinatewise, $\bm{b}_1\neq\bm{b}_0$, and $g_H(\bm{b}_0)>0$, then
\[
\mathcal{C}_{H}(\bm{b}_0)\subseteq\mathcal{C}_{H}(\bm{b}_1),
\]
and the two events differ with positive probability. Hence $g_H(\bm{b}_1)>g_H(\bm{b}_0)$. Theorem~1 therefore gives a finite MLE if and only if $\mathcal{S}$ has overlap.

Under overlap, the finite maximizer lies in $\mathcal{D}$, where strict concavity gives uniqueness on $V$. Since $L(\bm{\beta})=L(\operatorname{Proj}_{V}\bm{\beta})$, the ambient maximizer set is $\widehat{\bm{\beta}}+V^{\perp}$. A differentiable concave function has only maximizing critical points.
\end{proof}

The theorem combines two distinct forms of nondegeneracy. Threshold identifiability is a factor-level condition. It says that changing the effective threshold vector changes the corresponding latent event with positive probability. The spanning property is a parameter-level condition. It says that every nonzero direction in $V$ changes the threshold vector of at least one observed factor. Neither condition can replace the other. A probability may be strictly log-concave in its thresholds but remain constant along a parameter direction that all threshold maps annihilate. Conversely, a full-rank design cannot create strict curvature if a latent inequality is redundant.

Strictness passes to the full likelihood because distinct parameters in $V$ change at least one observed threshold vector. Threshold identifiability also gives strict threshold monotonicity, so Theorem~\ref{thm:main} supplies the existence criterion.

The corollaries provide simple checks for the main models. Full row rank covers choices, complete rankings, multivariate binary responses, sequential models, adjacent-category logit, and fixed-score stereotype models. Coordinate-bound matrices cover cumulative-link intervals and multivariate ordinal rectangles even though those matrices are not surjective. If $V=\R^p$, no ambient flat directions remain. In the choice case, Proposition~\ref{lem:rank-equivalence} shows that this is equivalent to full row rank of the combined design matrix $Z$.

\section[Separation in Ordinal Survey Data]{Separation in Ordinal Survey Data}
\label{sec:application}

This section illustrates with cumulative-link models how our results can be used in general statistical and data analytical practice. Separation in cumulative-link models was hitherto difficult to diagnose directly. The separation diagnostics use the R package \pkg{divoRce} \citep{separationcheck}.

We analyze an anonymized random sample of $n=300$ complete cases from a survey on willingness to share health data with a social health insurance provider in Austria \citep{drexler:2026}. Willingness (``WTSSHI'') was self-reported on an ordinal scale from 1 (lowest) to 4 (highest). Model predictors are respondents' self-reported trust in the provider (``trustSHI'') and self-reported knowledge about health data (``knowledge''), both treated as categorical. Figure~\ref{fig:wtsmosaic} shows a mosaic plot of the two predictors and the outcome. We specify an interaction of trust and knowledge and fit cumulative-link models with logit and complementary log-log (cloglog) links. 

\begin{figure}
  \centering
\includegraphics[width=.78\textwidth]{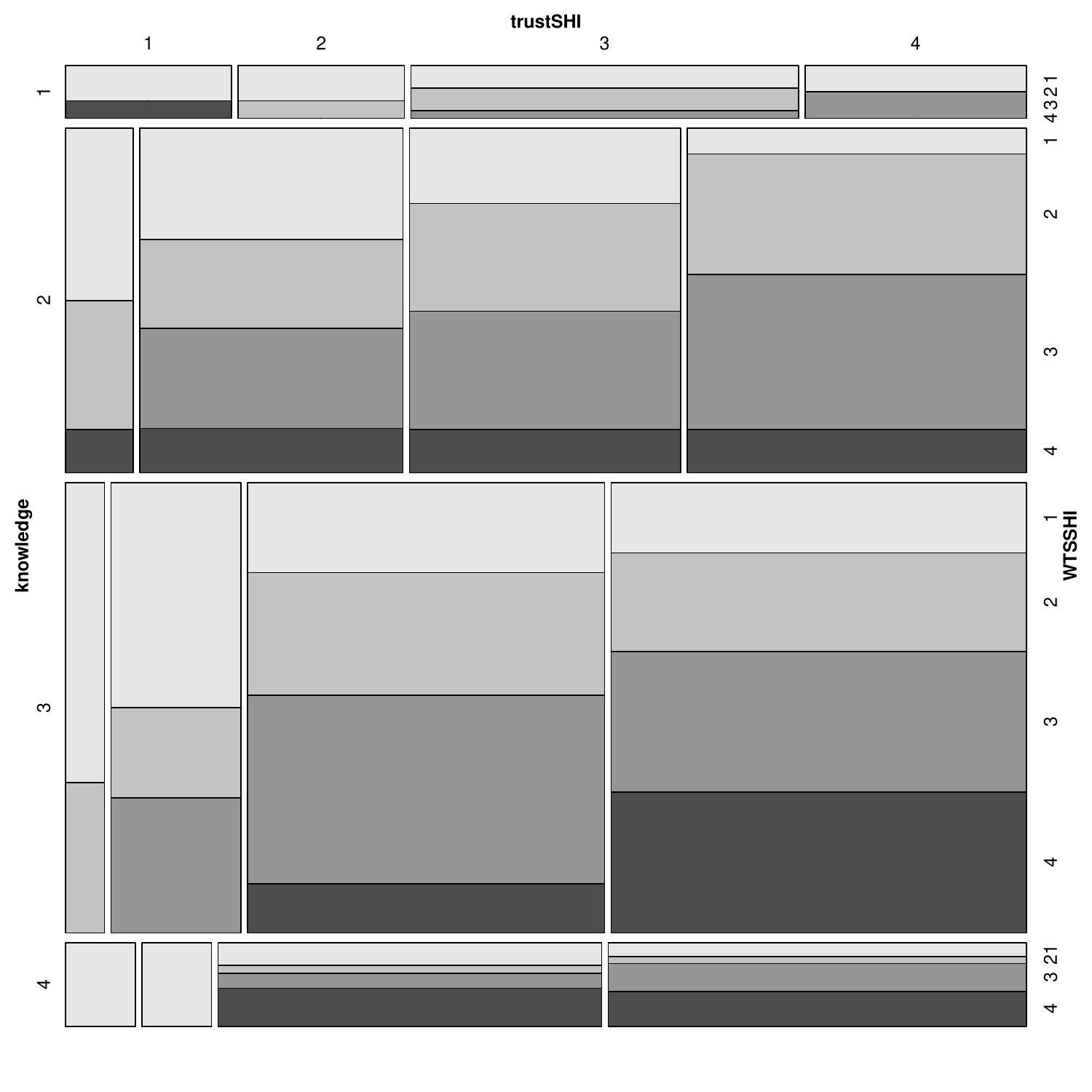}
\caption{Willingness to share health data (WTSSHI) by self-reported knowledge and trust in social health insurance provider (trustSHI). Respondents with knowledge level 4 (highest level) and trust levels 1 or 2 (lowest two levels) only reported willingness category 1 (lowest). This pattern is visible in the bottom left corner.}
\label{fig:wtsmosaic}
\caption*{\textit{Image description.} A mosaic plot displays willingness to share health data by knowledge and trust. Respondents with knowledge level 4 and trust levels 1 or 2 only reported willingness category 1.}
\end{figure}

Fitting these models produces a host of numerical symptoms: very large values for some coefficients, a numerically singular Hessian, and unavailable standard errors, Wald statistics, and $p$-values. The full coefficient tables are reported in the Appendix. These symptoms arise under both links. Their structural source is that respondents with knowledge level 4 and trust levels 1 or 2 report only the lowest willingness category, as visible in the lower left part of Figure~\ref{fig:wtsmosaic}. Changing the optimizer cannot restore a finite MLE for these specifications.

Based on the results in this article we can diagnose the situation in detail. Applying the linear program in Equation~\ref{lp2} confirms separation as the culprit, and the detailed diagnostics identify it as quasi-complete. The cone of separating directions from Section~\ref{sec:dim} has dimension $2$: two linearly independent directions weakly increase all observed thresholds and strictly increase some of them. This diagnosis is more informative than a generic convergence warning because it identifies both the type of failure and the multiplicity of the directions along which the likelihood can escape.

We find that the separation phenomenon is localized. Four respondents display a separating configuration. All four report knowledge level 4, trust levels 1 or 2, and willingness category 1. A columnwise variant of the linear program identifies four model terms associated with the nonexistence of a finite MLE: the indicator for knowledge level 4 and its interactions with trust levels 2, 3, and 4. Collapsing trust levels 1 and 2 does not resolve the problem: further checks confirm that separation persists after this recoding.

The results can be used to explore several possible responses to the diagnosis. First, a substantively justified model simplification may remove the interaction terms in favor of an additive model. Both additive specifications have overlap, and the likelihood-ratio test and AIC support simplification for the logit model but not for the cloglog model. Second, refitting after excluding the four respondents confirms the geometric diagnosis and provides a sensitivity analysis. The remaining data display overlap under both links, although two coefficients become aliased because the corresponding predictor combinations are absent. Exclusion from a primary analysis requires independent substantive justification. Additional sampling in these predictor configurations could also be considered, although restoring overlap depends on the outcomes observed. Third, when substantively appropriate, alternative model classes can be considered because their structure vectors may differ. Here, adjacent-category logit and ordered stereotype specifications are also flagged as separated. Fourth, penalized or bias-reduced procedures can provide finite alternatives \citep{firth,kosmidis-lunardon-2024}. The reproduction script illustrates bias reduction for the separated adjacent-category logit model retaining the trust by knowledge interaction; this however changes the response model from the cumulative-link specification.

The main practical value of our diagnostics is not to prescribe remedies to be automatically and blindly applied. Rather, it is to distinguish structural nonexistence from numerical failure and to identify the observations, model terms, and recession directions that require an analyst's or consultant's attention when applying categorical response models.

\section{Conclusion}
\label{sec:conclusion}

This paper reduces the existence and uniqueness theory of a broad collection of categorical response models to one geometric object, the pooled set of structure vectors of a latent polyhedral representation. Overlap of that set characterizes existence of a finite MLE, with sufficiency under continuity alone and necessity under strict threshold monotonicity. The argument dispenses with likelihood concavity, exchangeability, and full design rank, and rank deficiency is absorbed by the estimable span. Positive strictly log-concave latent densities with threshold-identifiable representations add uniqueness on that span. One linear program determines whether overlap holds for every member of the family, and the cone of separating directions measures the multiplicity of a failure.

For statistical practice and consulting, the framework turns numerical symptoms into a checkable property of the data and the model. It distinguishes structural nonexistence from numerical failure and identifies the observations and model terms associated with separation. Software developers can use these diagnostics to provide informative warnings when a finite MLE does not exist, together with details of the separating configuration. We recommend routinely using these checks as part of standard model diagnostics so that model interpretation, revision, penalized or bias-reduced estimation, or substantively justified data decisions proceed on clear and well-founded grounds.

The family boundary marks open problems. Reference-ratio models place comparison thresholds nonlinearly in the parameter, falling outside linear-threshold theory \citep{peyhardi+trottier+guedon:2015,glmcat}, though the affine case reduces to multinomial logit up to rescaling. Extending the overlap criterion to nonlinear threshold maps is the natural next step.

\section*{Data Availability Statement}
Data are available as \texttt{data(HDSS)} in the R package \pkg{divoRce} \url{https://cran.r-project.org/package=divoRce}. The R script to reproduce the results is provided as supplementary material and at \url{https://codeberg.org/trusch/divoRce_reproduction_script/src/branch/main/divoRce-reprodscript.R}.

\noindent\textbf{Funding.} Lukas Sablica's time during final preparation and submission was supported by funds of the Oesterreichische Nationalbank, Austrian Central Bank Anniversary Fund, project number 19079. The research underlying the manuscript was conducted before this funding.

\noindent\textbf{Disclosure statement.} The authors declare no competing interests.

\section*{Declaration of generative AI use}
The authors used OpenAI GPT 5.5 for spell-checking and language polishing and reviewed and edited the content. They take full responsibility for the content.

\bibliography{reference}

\clearpage
\appendix
\section{Additional Details for the Empirical Illustration}
\label{sup:illustration}

The coefficient tables below use cumulative-link models with a trust by knowledge interaction and no additional covariates. Table~\ref{tbl:wtsmod} records the numerical output for the two separated fits. These are optimizer stopping values, not finite MLEs. After excluding the four respondents identified by the diagnostics, both links have overlap. Table~\ref{tbl:wtsmod2} reports the cloglog sensitivity fit. Two coefficients are aliased after exclusion; the remaining coefficients are reported in the reduced parameterization chosen by \texttt{clm}. The reproduction script also fits additive, adjacent-category, and ordered stereotype alternatives.

\begin{table}[htbp]
\centering
\scriptsize
\begin{tabular}{lrrrrrrrr}
\toprule
 & \multicolumn{4}{c}{Logit} & \multicolumn{4}{c}{Cloglog} \\
\cmidrule(lr){2-5} \cmidrule(lr){6-9}
Variable & Estimate & Std. Error & z value & Pr($>|z|$) & Estimate & Std. Error & z value & Pr($>|z|$) \\
\midrule
trustSHI2 & -0.768 & NA & NA & NA & -1.687 & NA & NA & NA \\
trustSHI3 & 0.121 & NA & NA & NA & -1.056 & NA & NA & NA \\
trustSHI4 & 0.554 & NA & NA & NA & -0.572 & NA & NA & NA \\
knowledge2 & -0.001 & NA & NA & NA & -0.714 & NA & NA & NA \\
knowledge3 & -0.768 & NA & NA & NA & -1.687 & NA & NA & NA \\
knowledge4 & -20.896 & NA & NA & NA & -5.675 & NA & NA & NA \\
trustSHI2:knowledge2 & 1.734 & NA & NA & NA & 2.156 & NA & NA & NA \\
trustSHI3:knowledge2 & 1.128 & NA & NA & NA & 1.612 & NA & NA & NA \\
trustSHI4:knowledge2 & 1.101 & NA & NA & NA & 1.279 & NA & NA & NA \\
trustSHI2:knowledge3 & 1.717 & NA & NA & NA & 2.519 & NA & NA & NA \\
trustSHI3:knowledge3 & 2.014 & NA & NA & NA & 2.612 & NA & NA & NA \\
trustSHI4:knowledge3 & 2.238 & NA & NA & NA & 2.651 & NA & NA & NA \\
trustSHI2:knowledge4 & 0.768 & NA & NA & NA & 1.293 & NA & NA & NA \\
trustSHI3:knowledge4 & 23.071 & NA & NA & NA & 7.418 & NA & NA & NA \\
trustSHI4:knowledge4 & 22.892 & NA & NA & NA & 6.948 & NA & NA & NA \\
Threshold $1|2$ & 0.090 & NA & NA & NA & -1.442 & NA & NA & NA \\
Threshold $2|3$ & 1.354 & NA & NA & NA & -0.466 & NA & NA & NA \\
Threshold $3|4$ & 3.051 & NA & NA & NA & 0.573 & NA & NA & NA \\
\addlinespace
Log-likelihood & -381.142 & & & & -378.994 & & & \\
\bottomrule
\end{tabular}
\caption{\label{tbl:wtsmod} Numerical coefficient values returned by cumulative logit and cloglog fits with a trust by knowledge interaction ($n=300$). Neither model has a finite MLE. All standard errors, Wald statistics, and $p$-values are unavailable because the Hessian is numerically singular. The displayed values depend on the optimizer's stopping rule; they are not finite maximum likelihood estimates.}
\end{table}

\begin{table}[htbp]
\centering
\scriptsize
\begin{tabular}{lrrrr}
\toprule
Variable & Estimate & Std. Error & z value & Pr($>|z|$) \\
  \midrule
trustSHI2           &  $-1.687 $&   $ 0.960$ & $-1.758$ & $0.079$ \\ 
trustSHI3            & $-1.056$ &   $ 0.821$ & $-1.286$ & $0.198$  \\
trustSHI4            & $-0.572$ &   $ 0.893$ & $-0.640$ & $0.522$   \\
knowledge2           & $-0.714$ &   $ 0.807$ & $-0.885$ & $0.376$ \\  
knowledge3           & $-1.687$ &   $ 0.847$ & $-1.993$ & $0.046$ \\ 
knowledge4           & $ 1.273$ &   $ 0.664$ & $ 1.917$ & $0.055$ \\ 
trustSHI2:knowledge2 & $ 2.156$ &   $ 1.054$ & $ 2.046$ & $0.041$ \\ 
trustSHI3:knowledge2 & $ 1.612$ &   $ 0.928$ & $ 1.736$ & $0.082$ \\ 
trustSHI4:knowledge2 & $ 1.279$ &   $ 0.988$ & $ 1.294$ & $0.196$ \\ 
trustSHI2:knowledge3 & $ 2.519$ &   $ 1.094$ & $ 2.302$ & $0.021$ \\
trustSHI3:knowledge3 & $ 2.612$ &   $ 0.958$ & $ 2.727$ & $0.006$ \\
trustSHI4:knowledge3 & $ 2.651$ &   $ 1.020$ & $ 2.598$ & $0.009$ \\
trustSHI2:knowledge4 &      NA &        NA &     NA &      NA   \\
trustSHI3:knowledge4 & $ 0.471$  &  $ 0.880$ & $ 0.535$ & $0.593$  \\ 
trustSHI4:knowledge4 &      NA &        NA &     NA &      NA   \\
 Threshold $1|2$ & $-1.442$ & $0.719$ & $-2.007$ & $0.045$ \\
 Threshold $2|3$ & $-0.466$ & $0.714$ & $-0.653$ & $0.514$ \\
 Threshold $3|4$ & $0.573$ & $0.709$ & $0.808$ & $0.419$ \\
\addlinespace
Log-likelihood & $-378.994$ & & & \\
\bottomrule
\end{tabular}
\caption{\label{tbl:wtsmod2} Cumulative cloglog fit with the same interaction after excluding respondents $56$, $75$, $90$, and $131$ as a sensitivity analysis ($n=296$). Overlap holds. Two interaction coefficients are aliased and marked NA; the remaining estimates and standard errors refer to the reduced parameterization selected by \texttt{clm}. The full parameter vector is not uniquely identified.}
\end{table}

\end{document}